\documentclass[11pt]{article}
\usepackage{geometry}
\usepackage{graphicx}
\usepackage{amsmath}
\usepackage{amssymb}
\usepackage[all]{xy}
\usepackage{xspace}
\usepackage[tight]{subfigure}
\usepackage{mathrsfs}

\usepackage{xspace}  

\newtheorem{theorem}{Theorem}[section]
\newtheorem{corollary}[theorem]{Corollary}
\newtheorem{lemma}[theorem]{Lemma}
\newtheorem{proposition}[theorem]{Proposition}
\newtheorem{claim}[theorem]{Claim}
\newtheorem{definition}[theorem]{Definition}

\newtheorem{conjecture}[theorem]{Conjecture}

\newcommand{\ra}{\rightarrow}

\newcommand{\eps}{\varepsilon}
\newcommand{\xhdr}[1]{\paragraph{\bf #1}}
\newcommand{\omt}[1]{}

\newcommand{\R} {\ensuremath{\mathbb{R}}} 

\def\squarebox#1{\hbox to #1{\hfill\vbox to #1{\vfill}}}
\newcommand{\qed}{\hspace*{\fill}\vbox{\hrule\hbox{\vrule\squarebox{.667em}\vrule}\hrule}\smallskip}
\newenvironment{proof}{\noindent{\bf Proof:~~}}{\(\qed\)}

\newcommand{\w}{\alpha}         
\newcommand{\interact}{S}       
\newcommand{\influence}{T}      
\newcommand{\intgraph}{\mathcal {\interact} }  
\newcommand{\infgraph}{\mathcal {\influence} } 
\newcommand{\wn}{\mathcal{M}}   
\newcommand{\normalize}{\wn}   
\newcommand{\netflow}{flow\xspace} 
\newcommand{\netflows} {flows\xspace} 
\newcommand{\flow}[2]{f_{#1\to #2}} 
\newcommand{\ip}{p} 
\newcommand{\activegraph}{\infgraph_{\mathtt{act}}}  

\begin{document}

\title{Selection and Influence in Cultural Dynamics%
\footnote{%
A one-page abstract of this work has appeared in \emph{ACM Conf. on Electronic Commerce}, 2013.\newline
This work has been supported in part by
a Simons Investigator Award,
a Google Research Grant,
an ARO MURI grant,
and NSF grants
IIS-0910664, 
CCF-0910940, 
and
IIS-1016099. 
}
}

\author{
David Kempe
\thanks{
Department of Computer Science,
University of Southern California, Los Angeles CA 90089-0781, USA.
Email: dkempe@usc.edu.
}
\and
Jon Kleinberg
\thanks{
Department of Computer Science,
Cornell University, Ithaca NY 14853, USA.
Email: kleinber@cs.cornell.edu.
}
 \and
Sigal Oren
\thanks{Department of Computer Science,
Ben-Gurion University of the Negev, Beer Sheva 8410501, Israel.
This work has been done while S. Oren was a graduate student at Cornell University and a research intern at Microsoft Research. Email: sigalo@cs.huji.ac.il.}
 \and
Aleksandrs Slivkins
\thanks{
Microsoft Research, New York, NY 10011, USA.
Email: slivkins@microsoft.com.
}
}

\date{First version: April 2013\\This version: October 2015}

\begin{titlepage}
\maketitle

\begin{abstract}
One of the fundamental principles driving diversity or homogeneity in
domains such as cultural differentiation, political affiliation, and
product adoption is the tension between two forces: influence (the
tendency of people to become similar to others they interact with) and
selection (the tendency to be affected most by the behavior of others
who are already similar). Influence tends to promote homogeneity
within a society, while selection frequently causes
fragmentation. When both forces act simultaneously, it becomes an
interesting question to analyze which societal outcomes should be
expected.

To study this issue more formally, we analyze a natural
stylized model built upon active lines of work in political
opinion formation, cultural diversity, and language evolution.
We assume that the population is partitioned into ``types'' according to
some traits (such as language spoken or political affiliation). While
all types of people interact with one another, only people with
sufficiently similar types can possibly influence one another. The
``similarity'' is captured by a graph on types in which individuals of
the same or adjacent types can influence one another. We achieve an
essentially complete characterization of (stable) equilibrium outcomes
and prove convergence from all starting states.
We also consider generalizations of this model.
\end{abstract}

{\bf Keywords:} social networks, selection, influence, opinion formation.

\thispagestyle{empty}
\end{titlepage}



%
%
%
%


\section{Introduction} \label{sec-intro}
\xhdr{Selection and Influence.}
Human societies exhibit many forms of cultural diversity ---
in the languages that are spoken, in the opinions and values that
are held, and in many other dimensions.
An active body of research in the mathematical social sciences has
developed models for reasoning about the origins of this diversity,
and about how it evolves over time.

One of the fundamental principles driving cultural diversity is
the tension between two forces: influence and selection.
\emph{Influence} refers to the tendency of people to become similar to
those with whom they interact, whereas \emph{selection} (or
  \emph{choice homophily} \cite{mcpherson-homophily}) is the tendency
of people to interact with those who are more similar to them, and/or
to be more receptive to influence from those who are similar.%
\footnote{While \emph{selection} may sometimes have causes other than similarity, such as attraction of the opposites or triadic closure, we focus on similarity-driven selection throughout this paper.
We use the term {\em selection} rather than {\em homophily} because the latter is sometimes used to refer to the broader fact that people tend to be similar to their neighbors in a social network, regardless of the mechanism leading to this similarity.}

Both of these forces lead toward outcomes in which people
end up interacting with others like themselves, but
in different ways:
influence tends to promote homogeneity, as people shift their
behaviors to become alike, while
selection tends to promote fragmentation, in which a society can
split into multiple groups that have less and less interaction with
each other.
Research that offers qualitative analyses for issues such as
consensus-building, political polarization, or social stratification
can often be interpreted through the lens of this
influence-selection trade-off
\cite{cohen-peer-group,kandel-socialization,mcpherson-homophily}.
The trade-off between influence and selection, and the
development of data analysis techniques to try separating
the effects of the two,
have been integral to understanding and promoting the adoption of
products and behaviors in social networks
\cite{anagnostopoulos-selection-influence,aral-selection-influence,bramoulle-peer-effects,lafond:neville:tests,shalizi-selection-influence},
an active line of work at the interface of computing, economics, and statistics.

When both influence and selection are operating at the same time, how should
we reason about their combined effects?
In particular, as Axelrod~\cite[p.~203]{axelrod-culture} asked:
\begin{quote}
If people who are similar to one another tend to become more alike in their
beliefs, attitudes, and behavior when they interact, why do not all such
differences eventually disappear?
\end{quote}


Several lines of modeling work have approached this question, all starting
from similar underlying motivations, but developing different mathematical
formalisms.

\begin{enumerate}
\item Research on political opinions has studied populations in which
each person holds an opinion. The opinion is represented by a number
drawn from a bounded interval on the real line $\R^1$, or from a discrete
set of points in an interval.
(For example, the interval may represent
the political spectrum from liberal to conservative.)
Each person is influenced by the opinions of others
who are sufficiently nearby on the interval, thus capturing the
interplay between influence (people are shifting their opinions
based on the opinions of others) and selection (people only pay
attention to others whose opinions are sufficiently close)
\cite{ben-naim-opinions,deffuant-opinions,hegselmann-opinions}.
\item Axelrod proposed a model of cultural diversity in which
there are several {\em dimensions} of culture, and each person
has a value associated with each dimension (e.g., a choice of
language, religion, or political affiliation).
Agents are more likely to interact when they agree on more
dimensions; when two people interact, one person randomly
chooses a dimension in which they differ, and changes his value
so that they now match in this dimension \cite{axelrod-culture}.
For example, two people who have passions for similar sports and
styles of food may end up having an easier time
(and more opportunity for) associating, and
hence an easier time influencing one another along another
dimension such as religious beliefs.
Again, the model represents an influence process in which the
interactions are governed by selection based on (cultural) similarity.
Axelrod's model has generated a large amount of subsequent work;
see \cite{castellano-sociophysics} for a survey.
\item Finally, Abrams and Strogatz exhibited some of the interesting
effects that can occur even when there are only two types of people.
They modeled a scenario in which people speak one of two
languages. People mainly interact with speakers of their own
language, but there is gradual ``leakage'' over time as speakers of one
language may convert to become speakers of the other
\cite{abrams-strogatz-language}.
The Abrams-Strogatz model has also generated an active line of
follow-up results, including explorations of its
microfoundations through agent-based simulation
\cite{stauffer-micro-abrams-strogatz} and analyses of the
spatial effects and population density \cite{patriarca-language-competition}.
\end{enumerate}

\xhdr{Commonalities among Models.}
Although the models described above differ in many details, they have
the same underlying structure: the population is divided into
a set of {\em types} (the opinions, the cultural choices,
the language spoken), and a person of any given type may be
influenced to switch types, but only by others whose types
are sufficiently similar.
(In the case of the Abrams-Strogatz model, there is a preference
for one's own type, but since there are only two types,
all types can influence each other.)
This process generates a ``flow'' as people migrate
among different types, and we can ask questions about both dynamics
(which outcomes the process will reach) and equilibria
(which outcomes are self-sustaining, in the sense that
the flows between types preserve the fraction of people
who belong to each type).
Following the language around Axelrod's work,
we will refer to this type of process as representing the
{\em cultural dynamics} of the population.

In addition to their similarities in structure, these
cultural dynamics models also agree in their broad conclusions.
In the first two models, the population
gradually separates into distinct ``islands''
in the space of possible types;
subsequently, no further interaction between the islands is possible.
In the Abrams-Strogatz model, with just two types,
the only outcomes that are stable under perturbations are
the two extreme outcomes in which everyone ends up belonging to the
same type.
Typically, there is also an unstable equilibrium in which
each language is spoken by a non-zero fraction of the population.

The most salient difference among the models is the structure
that is imposed on the set of types.
In each case, there is an undirected {\em influence graph} $\infgraph$ on the set of types:
when a person of type $u$ interacts with a person of type $v$,
the person of type $u$ has the potential to switch to
{(or move towards)} $v$ {\em provided}
that $u$ and $v$ are neighbors in $\infgraph$ (i.e., provided that $u$ and $v$
are sufficiently similar according to the interpretation of the model).
In the models of one-dimensional opinion dynamics on a discrete set,
the graph $\infgraph$ is the $k^{\rm th}$ power of a path for some $k \geq 1$
(types are similar enough when they are within $k$ steps on the path);
in Axelrod's model, the graph $\infgraph$ is the $k^{\rm th}$ power of a
(not necessarily binary) hypercube.
The Abrams-Strogatz model shows that these kinds of processes can exhibit
subtle behavior even on a two-node influence graph $\infgraph$.

\subsection*{The present work:
Cultural Dynamics on an Arbitrary Influence Graph}
All of the prior results apply only to highly structured,
symmetric graphs (essentially hypercubes and paths),
whereas in some of the settings that the models seek to capture, the
set of types can have a less orderly structure.
As one simple example, consider a subgraph of the
  ``religion'' graph (depicted in Figure~\ref{fig:religion}), with the
following 6 types:
agnostics (AG), atheists (AT), casual protestants (CP), devout
protestants (DP), casual catholics (CC), and devout catholics (DC).
Here, it is reasonable to assume that transitions happen between the
casual versions of each belief, or between casual and devout
versions of the same belief. In other words, the graph would consist
of a triangle AG-CP-CC, and edges AT-AG, DP-CP, DC-CC.\footnote{%
We may expect transitions between other states to happen, albeit with
much smaller probability. We will discuss this issue more in
Section~\ref{sec:modeling-choices}.}

\begin{figure}[htb]
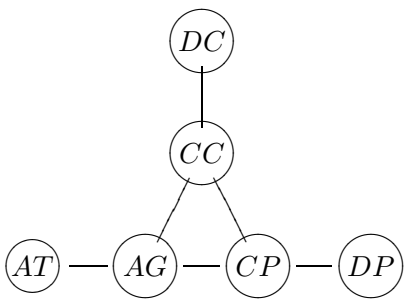

\begin{align*}
\xygraph{ !{<0cm,0cm>;<1.5cm,0cm>:<0cm,1.5cm>::} !{(1,0)}*+[o]+[F]{AT}="1" !{(2,0) }*+[o]+[F]{AG}="2" !{(3,0) }*+[o]+[F]{CP}="3" !{(4,0)}*+[o]+[F]{DP}="4" !{(2.5,1) }*+[o]+[F]{CC}="5" !{(2.5,2) }*+[o]+[F]{DC}="6" "1"-"2" "2"-"3" "3"-"4" "2"-"5" "3"-"5" "5"-"6"}
\end{align*}
\caption{\label{fig:religion}
A small subgraph of the ``religion'' graph, showing edges along which
transitions may happen between agnostics, atheists, casual
protestants, devout protestants, casual catholics, and devout catholics.
}

\end{figure}



To ensure that insights derived from the analysis of a model (such as
the ones for hypercubes and graphs) are not limited to those specific
models, and to further understand the governing principles, it is
desirable to understand the dynamics and equilibria of the process in
more general graphs.

This is the problem we address in the present work,
where we develop techniques for resolving some of the main questions
on arbitrary graphs, under a clean and stylized model of interactions.
For a natural formulation of cultural dynamics on an arbitrary
influence graph (which we refer to as the \emph{global model},
for reasons explained later),
we prove convergence results and precisely characterize
the set of all stable equilibria.
We then consider generalizations of the global model,
extending some of our convergence and stability
results to these more general settings
and posing several open questions.

\xhdr{The Global Model.}
We now describe the global model in more detail.
Because the models from the earlier lines of work discussed above differ
in many of their details, there is no meaningful way to
simultaneously generalize all of them in a precise syntactic sense.
Instead, our goal is to formulate a version of cultural dynamics
that exhibits the same basic interplay of selection and influence ---
specifically, the idea that influence only happens among types that
are ``close together'' --- while allowing for an arbitrary graph
on the set of types.

Let $\infgraph$ be a graph on a finite set of types $V$
of cardinality $n = |V|$;
for each type $u \in V$, let $\influence_u \subseteq V$
denote the set of $u$'s neighbors in $\infgraph$.
As is standard in many of the approaches to cultural dynamics, we model
the population as a continuum\footnote{This and other
    modeling choices are discussed in more detail in
    Section~\ref{sec:modeling-choices}.}:
at the start of the process, each type $u \in V$ has a non-negative
population {\em mass} associated with it, corresponding to the
fraction of the population that initially has this type.
(Consider, for example, the fraction of the world's population
that belongs to a certain religion or speaks a certain language.)
Time evolves continuously\footnote{Again, refer to
    Section~\ref{sec:modeling-choices} for a discussion.}
and $x_u(t)$ denotes the mass on type $u$ at time $t$.
The full state of the population at time $t$ is thus given
by the {\em mass vector} $x(t)$, the vector of values $x_u(t)$
for all $u \in V$.

We define a continuous-time dynamical system in which
the direction in which the populations move
is determined in terms of the mass vector $x(t)$.
The dynamical system is motivated by imagining that each person
chooses a random other person to interact with.
Selection effects are captured in two ways by the model:
first, people are more likely to interact with their type;
second, they only have the potential to be influenced
when they interact with an individual of their own or a neighboring type.
Specifically, each person is $\w$ times more likely to choose
an interaction partner of their own type than someone of a different type,
for a parameter $\w \geq 1$.
When a person of type $u$ chooses to interact with a person of type $v$,
such that $v \in \influence_u$, with probability $\ip$, he will
switch to type $v$, where $\ip\in (0,1]$ is a fixed parameter.
To express this dynamic numerically, we let
$\normalize_u(t) = \w x_u(t)+\sum_{v \in V \setminus \{u\}}x_v(t)$.
A person of type $u$ chooses to interact with his own type with probability
$\w x_u(t) / \normalize_u(t)$,
and with any other type $v \neq u$ with probability
$x_v(t) / \normalize_u(t)$.
Thus, the fraction of the \emph{entire population} which is
moving from $u$ to $v$ is $\ip \cdot x_u(t) x_v(t)/\normalize_u(t)$.
At the same time as this mass of $\ip \cdot x_u(t) x_v(t)/\normalize_u(t)$ is moving from $u$ to $v$,
a mass of ${\ip \cdot}x_v(t) x_u(t)/\wn_v(t)$ is moving from $v$ to $u$.
These movements partially cancel each other out, and motivate the
following definition of the (directed) \emph{\netflow{}} on the edge $(v,u) \in \infgraph$:
\begin{align}\label{eq:model-flows}
\flow{v}{u}(t) &= \ip \cdot x_v(t)\, x_u(t)\, \left( \frac{1}{\wn_v(t)} - \frac{1}{\wn_u(t)} \right).
\end{align}
The change in mass at a node $u$ can then be written as
\begin{align}\label{eq:def-by-flows}
\dot{x}_u(t) & = \sum_{v\in \influence_u} \flow{v}{u}(t)
   \; = \; \ip \cdot x_u(t)\, \sum_{v\in \influence_u}
       x_v(t)\, \left( \frac{1}{\wn_v(t)} - \frac{1}{\wn_u(t)} \right).
\end{align}

Notice that because the system is characterized by a system of
differential equations and that for all $u$ the derivative of
$x_u(t)$ is finite (by Equation \ref{eq:def-by-flows}), we obtain that $x_u(t)$ is continuous for all
$u$.

%

It is natural to think of the parameter $\ip$ as generally being very small,
since most interactions between people do not lead to a change of type.
However, as it turns out, the value of $\ip$ does not have a major
qualitative effect on our results. This is not surprising, since
introducing $p<1$ (as opposed to $p=1$) just slows down the flow
between any two types by a factor of $1/\ip$. We include $\ip$ in the
model in order to capture the range of possible speeds at which
transitions can happen.  For example, if the types in
our model correspond to dialects of a language, we can choose a small
$\ip$ (since the probability that a person changes his dialect is
very small). However, if instead the types represent opinions in the period
before an election, people may switch much more rapidly, and a larger
$\ip$ is appropriate.

\xhdr{Convergence, Equilibria, and Stability in the Global Model.}
Our first result is that for any influence graph $\infgraph$
and any initial mass vector $x$
 the system converges
to a limit mass vector $x^*$.
We prove this by establishing a system of invariants on the population
masses over time, capturing a certain ``rich-get-richer''
property of the process --- essentially, that the types of large mass
will tend to grow at the expense of the types of small mass.

We next consider the equilibria of this model: we say that
a mass vector $x$ is an equilibrium if it remains
unchanged after one application of the update rule.
It is easy to construct examples of equilibria that are not stable,
in the sense that an arbitrarily small perturbation of the
masses $x_u^*$ can --- after further applications of the update rule ---
push the masses far away from the equilibrium.
Such equilibria are less natural as predicted outcomes of the
cultural dynamics being modeled, since the population would
be unlikely to hold its position near this equilibrium.

To make this statement precise, we use the notion of {\em Lyapunov stability}.
We say that an equilibrium $x^*$ is {\em Lyapunov stable}
if
given any $\eps>0$, there exists a $\delta > 0$ such that if
$||x(t_0)-x^*||_1 < \delta$, then $||x(t)-x^*||_1 < \eps$,
for all $t \geq t_0$.\footnote{
One could ask about stronger notions of stability,
in particular, {\em asymptotic stability},
which requires that there exists a $\delta_1>0$ such that if
$||x(t_0)-x^*||_1 < \delta_1$, then $x(t) \rightarrow x^*$
as $t \rightarrow \infty$.
Asymptotic stability is not a useful definition for our purposes;
for example, if the underlying influence graph $\infgraph$ has no edges,
then any assignment of population masses is an equilibrium, but none
are asymptotically stable, since there is no way for a small
perturbation to converge back to the original state.
On the other hand, all equilibria are Lyapunov-stable in this
simple example.}
For simplicity, we use the $L_1$ norm
throughout. Since our vectors have finite dimensionality, the
  different $L_p$ norms only differ by constant factors, so all stability
  results apply to other norms by scaling $\eps, \delta$ appropriately.

We prove that $x^*$  is a Lyapunov-stable equilibrium if and only if
the set of active types $A(x^*) = \{u : x_u^* > 0\}$
is an independent set in the influence graph $\infgraph$.
The proof is based on the rich-get-richer properties
of the process; these properties are used to show
that after a sufficiently small perturbation to the population masses, the
amount by which any type with positive mass can grow is bounded.

\xhdr{Interpretations of the Basic Results.}
The basic results discussed above establish
a precise sense in which the natural equilibria
tend to break the population into non-interacting islands.

In addition to offering a qualitative statement about fragmentation
of opinions under the proposed stylized model,
the results also suggest a way of reasoning about the
phenomenon by which opinions on different issues
tend to become aligned, with an individual's views on one issue
providing evidence for his or her views on another
\cite{poole-one-dim-opinions,spector-one-dim-debate}.
To take a concrete example that already appears on the 2-dimensional
hypercube (i.e., the 4-node cycle), consider a setting in which
each individual has either a liberal or conservative view on
fiscal issues and either a liberal or conservative view on social issues.
If we assume that people only influence each other when they agree
on at least one of these two categories of issues, then the graph on the set
of types is a 4-node cycle.
Since our results on Lyapunov-stable equilibria indicate that independent sets
are favored as outcomes, we can interpret the conclusion in this example
as predicting that under the proposed model, either the
whole population will converge on a single node (representing a
uniform choice of views), or on an independent set of two nodes, in
which case an individual's opinion on fiscal issues has become
correlated with his or her opinion on social issues.


It is also instructive to compare our results to the main result
of \cite{abrams-strogatz-language} discussed above.
Recall that they consider
the influence graph $\infgraph$ $= K_2$ (two connected nodes),
and they find that the two stable equilibria are
the outcomes in which all the population mass is
gathered at a single node.
The family of dynamical systems they consider strictly subsumes ours
in the special case of a two-node graph, but for the specific system
we study, our results imply that
their basic finding extends to arbitrary graphs:
in any graph, the Lyapunov-stable equilibria correspond to the non-empty
independent sets, just as Abrams and Strogatz showed for the
two-node graph $K_2$.

\xhdr{A Generalization: Limiting both Interaction and Influence.}
We now discuss a natural generalization of the model that is
significantly more challenging to analyze.
In the global model, the members of type $u$
can interact with members of \emph{all} other types,
even though they are influenced only by the types in
$\influence_u$.
However, there are settings in which it is more natural
to assume that the members of a type only ever {\em interact} with members
of a subset of the other types; for example, this may be a reasonable
assumption when types represent different languages.
This is somewhat similar to the approach Centola et
al.~\cite{centola2007homophily} in studying a variant of the Axelrod model.
Under this variant, there is a social network among the agents;
whenever two agents become so different that they cannot influence one
another any more, the tie between them is broken, and new ties are formed.
Centola et al.~\cite{centola2007homophily} use simulations to show
that multi-cultural equilibria form readily and stably under this model.

To capture the idea that some types may simply be too
  different to interact, we now assume that there are two potentially
distinct graphs on the set of types $V$:
the influence graph $\infgraph$ (as before), as well as
an undirected {\em interaction graph} $\intgraph$,
where $\infgraph$ is a subgraph of $\intgraph$.
Rather than interacting with a person chosen from the full population,
a member of type $u$ selects an interaction partner
from the set $\interact_u$ of $u$'s neighbors in $\intgraph$.
It is straightforward to write the new update rule for this
more general dynamical system, by summing over
types in $\interact_u$ instead of $V \setminus \{u\}$.
Specifically, we can now define
\begin{align}
\normalize_u(t) & = {\w x_u(t)+\sum_{v \in \interact_u}x_v(t)}.
\label{eq:generalized-interaction-mass}
\end{align}
With this new definition of $\normalize_u(t)$, the
  definitions of flows and updated masses at nodes from
  Equations~\eqref{eq:model-flows}, \eqref{eq:def-by-flows}
  stay exactly the same.
Hence, the $\wn_{u}(t)$ terms emerge
as crucial quantities that determine the direction of the \netflow;
for that reason, we will call $\wn_u(t)$ the \emph{interaction mass} of
node $u$ at time $t$.


The global model is simply the special case
in which the interaction graph $\intgraph$ is the complete graph.
The name \emph{global model} emphasizes that
each type interacts ``globally,'' with all other
types.\footnote{
There are clearly many other potential generalizations
which could incorporate notions of non-uniform interaction, including
different interaction strengths between different pairs of types.
Such extensions would lead to interesting questions as well.
In the present work, we focus on the generalization with two unweighted graphs
$\intgraph$ and $\infgraph$ because it captures in a direct way
some of the additional complexity that is introduced
by simultaneously modeling limited interaction and influence.}

The behavior of this general model is significantly more complex
than the behavior of the global model; for instance, for arbitrary
$\intgraph$ and $\infgraph$, it is not even clear whether the process
will always converge.
Intuitively, much of the difficulty comes from the fact that when
we consider two neighboring types $u$ and $v$, the sets of types that
they are interacting with, $\interact_u$ and $\interact_v$, can be
quite different, whereas in the global model they are both the full set $V$.
Among other things,
this can lead to violations of the rich-get-richer property
that was so useful for reasoning about the dynamics of the global model.

For the general model, we first establish a necessary condition for
equilibria, as well as sufficient conditions for convergence and stability.
We then focus further on the special case in which
$\intgraph = \infgraph$.
This is in a sense the opposite extreme from the global model;
instead of making $\intgraph$ as large as it can be, we make it
as small as possible subject to the constraint that it contains
$\infgraph$ as a subgraph.
Accordingly, we refer to the case $\intgraph = \infgraph$
as the {\em local model}.
There are many interesting open questions surrounding the
behavior of the local model; we make progress on these through
initial convergence results and the identification of a
large class of equilibria that are Lyapunov-stable for all $\w>1$: non-empty independent sets for
which all nodes in the set are at a mutual distance of at least three.
In fact, this is an ``if and only if'' characterization for an
important class of influence graphs: those whose connected components
are trees or, more generally, bipartite graphs.

An interesting observation is that
the local and global models can have genuinely different
behaviors starting from the same initial conditions:
Figure~\ref{fig:models-example} shows an example of an
initial mass distribution on the three-node path for which
the global model converges to an outcome in which the mass
is divided evenly between the two endpoints, while
the local model converges to the outcome in which
all the mass is on the middle node.

\begin{figure}[htb]
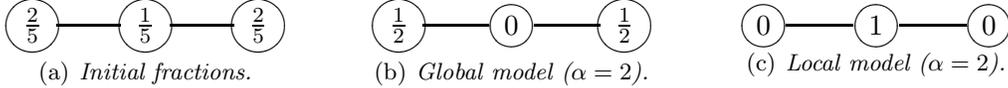

\begin{center}
\subfigure[\emph{Initial fractions.}]{
\xygraph{ !{<0cm,0cm>;<1.5cm,0cm>:<0cm,1.5cm>::} !{(1,0) }*+[o]+[F]{\tfrac{2}{5}}="1" !{(2,0) }*+[o]+[F]{\tfrac{1}{5}}="2" !{(3,0) }*+[o]+[F]{\tfrac{2}{5}}="3" "1"-"2" "2"-"3"}
 \label{fig:init-op}
} \hspace{7mm}
\subfigure[\emph{Global model ($\w=2$).}]{
\xygraph{ !{<0cm,0cm>;<1.5cm,0cm>:<0cm,1.5cm>::} !{(1,0) }*+[o]+[F]{\tfrac{1}{2}}="1" !{(2,0) }*+[o]+[F]{0}="2" !{(3,0) }*+[o]+[F]{\tfrac{1}{2}}="3" "1"-"2" "2"-"3"}
 \label{fig:nash-op}
} \hspace{7mm}
\subfigure[\emph{Local model ($\w=2$).}]{
\xygraph{ !{<0cm,0cm>;<1.5cm,0cm>:<0cm,1.5cm>::} !{(1,0) }*+[o]+[F]{0}="1" !{(2,0) }*+[o]+[F]{1}="2" !{(3,0) }*+[o]+[F]{0}="3" "1"-"2" "2"-"3"}
 \label{fig:optimal-op}
}
\caption{
{\small
An instance in which different models predict convergence to different
equilibria. The global model predicts an outcome in which two
non-interacting types survive (polarization), whereas the local model
predicts that only a single types survives (consensus).
\label{fig:models-example}
}
}
\end{center}
\vspace*{-0.2in}
\end{figure}

At a higher level, formalizing the distinction between
interaction ($\intgraph$) and influence ($\infgraph$) is a potentially
promising activity more broadly, particularly
in light of the considerable recent
interest in the effects of information filtering on the political process.
(See \cite{pariser-filter-bubble,sunstein-republic2} for
popular media accounts, and \cite{bakshy-socnet-info-diffuse}
for recent experimental research.)
The concern expressed in all these lines of work is that
personalization on the Internet makes
it possible to sharply restrict the diversity of information one sees,
and thus risks accentuating the degree of polarization and fragmentation
in political discourse --- essentially, the risk is that
people will only ever be
exposed to those who already agree with them, making any kind of
consensus almost impossible to achieve.

In this context, our general model also brings into the discussion the
interesting contrast between interaction and influence.
Personal filtering of information by Internet users can
restrict the set of people they interact with (affecting the sets
$\interact_u$), and it can also, separately, restrict the set of people
who may be able to influence them (affecting the sets $\influence_u$).
These two different effects are often bundled
together in discussions of information filtering; it will
be interesting to see whether treating them as genuinely distinct
can shed additional light on this set of issues.

\xhdr{Additional related work.}
Steglich et al.~\cite{Steglich:Snijders:Pearson} (see also \cite{Snijders:Steglich:Schweinberger}) provide a general model of social networks that combines selection and influence. While we focus on deriving structural properties of a network, these papers pursue a different goal: statistically valid inference of network parameters from real-life observations. More broadly, inferring latent properties of a social network from observations has been an active line of work. Some of the notable directions in this work, aside from the one taken in \cite{Steglich:Snijders:Pearson,Snijders:Steglich:Schweinberger}, include latent ``social space" reconstruction
 (e.g., \cite{handcock:raftery:tantrum,hoff:raftery:handcock}) and community detection (e.g., see \cite{brandes:erlebach:network-analysis,schaeffer:graphs-clustering-survey,fortunato:community-survey}). 


\section{Observations on the General Model}
\label{sec-general}

In this section, we develop several observations that apply to
the fully general model with an arbitrary interaction graph $\intgraph$.
In the subsequent sections, we utilize these
observations to analyze the global model
(where $\intgraph$ is the complete graph)
and the local model (where $\intgraph = \infgraph$).

We say that a node $u$ is {\em active} at time $t$ if
$x_u(t) > 0$, and {\em inactive} if $x_u(t) = 0$. We
occasionally refer to a node $u$ as \emph{$x$-active} if it is active in $x$ and
\emph{$x$-inactive} otherwise.
The set of all active nodes under $x$ is denoted by $A(x)$.
Much of our analysis concerns the structure of the
subgraph $\activegraph(x)$ of the influence graph $\infgraph$
induced by the active nodes $A(x)$.
We begin by characterizing when mass vectors are
in equilibrium.

\begin{proposition} \label{prop_eq}
A vector $x^*$ is an equilibrium if and only if each connected component
$C$ of $\activegraph(x^*)$ has the property that all nodes $u \in C$
have the same interaction mass $\wn_u$.
\end{proposition}

\begin{proof}
$x^*$ is at equilibrium if and only if the \netflow on all
edges is 0. In turn, from Equation~\eqref{eq:model-flows}, we see that
the \netflow on the edge $(u,v)$ is 0 if and only if at least one of
the following two conditions holds:
(1) $\wn_u = \wn_v$, (2) $x^*_u \cdot x^*_v = 0$.

If $x^*$ satisfies the assumptions, then each edge $(u,v)$ is
either inside a component (and thus $\wn_u = \wn_v$) or has at least
one inactive endpoint (and thus $x^*_u \cdot x^*_v = 0$).
Conversely, if $x^*$ is an equilibrium, each edge satisfies (1) or (2).
When $u,v$ lie in the same component $C$, there is a path between
them in $C$, and along that path, (1) must hold for all edges, so $u$
and $v$ must have the same interaction mass.
\end{proof}

The following useful lemma relates convergence and the change in
directions of \netflows:
\begin{lemma} \label{lem_gen_conv}
If there exists a time $t_0$ such that
the \netflows do not change direction after time $t_0$,
then the system converges.
\end{lemma}

\begin{proof}
Let $G$ be the directed graph obtained by directing each edge $(u,v)$ of
$\infgraph$ according to the direction of the corresponding \netflow
$\flow{u}{v}(t_0)$.
By the assumption, these directions stay constant after time $t_0$.
As flow always goes from types with smaller interaction mass
to types with larger interaction mass, $G$ must be acyclic.
Let $v_1, v_2, \ldots, v_n$ be a topological sorting of the graph,
so that all directed edges of $G$ are of the form $(v_i, v_j), i < j$.

We define $X_k(t) = \sum_{i=1}^k x_{v_i}(t)$ to be the total mass at
time $t \geq t_0$ on the $k$ first nodes in the topological sorting.
Because the total mass in the system is constant,
and all \netflow goes from nodes
with lower indices to nodes with higher indices, each of the
$X_k(t)$ must be non-increasing as a function of $t$. Since they are
also lower-bounded by 0, each $X_k(t)$ must converge to some value
$Z_k$ as $t \to \infty$. Therefore, each $x_{v_i}(t)$ converges to
$Z_i-Z_{i-1}$ as $t \to \infty$.
\end{proof}

Recall that we are interested in characterizing Lyapunov-stable equilibria.
We next provide a sufficient condition.

\begin{proposition} \label{prop_gen_stb}
An equilibrium $x^*$ is Lyapunov-stable if it satisfies the following two properties:
\begin{enumerate}
\item The active nodes form an independent set in the influence graph $\infgraph$.
\item The interaction mass of every active node is strictly greater
  than the interaction mass of each of its inactive neighbors
  in the influence graph $\infgraph$.
\end{enumerate}
\end{proposition}

\begin{proof}
Let $x^*$ be an equilibrium for which both properties hold.
A node $u$ is called \emph{$x^*$-active} if it is active in $x^*$ and
\emph{$x^*$-inactive} otherwise.
Let $A$ be the set of all $x^*$-active nodes, and let
$\wn^*_u$ denote the interaction mass of node $u$ with respect to $x^*$.
Define
\begin{align}
\delta & =  \tfrac{1}{{2\w+1}} \;
            \min_{u \in A,\; v \notin A,\; (u,v) \in \infgraph}\; (\wn^*_u - \wn^*_v)
       \; > \; 0
\label{eqn:delta-definition}
\end{align}
by the second property of $x^*$.
To show stability, we prove that whenever
$||x^*-x(t_0)||_1 \leq \delta$, the system will satisfy
$||x^*-x(t)||_1 \leq \delta$ for all times $t \geq t_0$.

The key step of the proof is to establish that for each node
$u \in A$, the mass $x_u(t)$ is non-decreasing over time,
i.e., that $\dot{x}_u(t) \geq 0$ for all $t \geq t_0$.
The initial condition implies that
$x_u(t_0) \geq x^*_u - \delta$ for all $u \in A$,
$\sum_{u \in A} x_u(t_0) \geq \sum_{u \in A} x^*_u - \delta$ and
$\sum_{v \notin A} x_v(t_0) \leq \delta$.
We will show that these invariants are maintained for all $t \geq t_0$.

Consider any time $t \geq t_0$ and edge $(u,v)$ with $u \in A$ and $v
\notin A$. The invariants imply that
$\wn_u(t) \geq \wn^*_u - \w \delta$ and
$\wn_v(t) \leq \wn^*_v + \w \delta$,
so we obtain that
$$\wn_u(t) - \wn_v(t) \geq (\wn^*_u - \wn^*_v) - 2\w \delta
 \overset{\eqref{eqn:delta-definition}}{\geq}
(2\w+1) \delta - (2\w) \delta > 0.$$
In particular, this implies that $\flow{v}{u}(t) \geq 0$; because this
holds for all $v \notin A$, we have established that
$\dot{x}_u(t) \geq 0$. By summing over all nodes $u \in A$, we have
shown that the invariant continues to hold.

Finally, because each of the $x_u(t), u \in A$ is non-decreasing,
mass can only move among $x^*$-inactive nodes, or from $x^*$-inactive
nodes to $x^*$-active ones. Therefore,
$\sum_{u \in A} |x_u(t)-x_u(t_0)|
= \sum_{u \notin A} x_u(t_0) - \sum_{u \notin A} x_u(t)$.
Thus,
\begin{align*}
{||x(t)-x^*||_1}
 & \leq \sum_{u \in A} |x_u(t)-x_u(t_0)| + \sum_{u\in A} |x_u(t_0)-x^*_u|
       + \sum_{u \notin A} x_u(t) \\
 & = \sum_{u \in A} |x_u(t_0)-x^*_u| + \sum_{u \notin A} x_u(t_0) = ||x(t_0)-x^*||_1 \; \leq \; \delta,
\end{align*}
so the system is Lyapunov-stable.
\end{proof}


\section{The Global Model} \label{sec-global}
In this section, we analyze the global model.
The definition of the general model states
that flows are always directed from nodes with smaller
interaction mass to nodes with larger interaction mass.
For the global model, this property is simplified significantly:
flow is always directed from types with smaller mass to types with
larger mass.
This property lets us achieve an almost complete understanding of the
global model.
We show that for this model, the system always converges, and we
present a complete characterization of which equilibria are Lyapunov-stable.
First, we characterize equilibria by applying
Proposition~\ref{prop_eq} to the global model.

\begin{corollary} \label{cor_equal-mass}
Under the global model with $\w > 1$, the system is at
equilibrium $x^*$ if and only if the following holds:
for every connected component $C$ of $\activegraph(x^*)$,
all nodes $u \in C$ have the same mass.
\end{corollary}

\begin{proof}
Proposition~\ref{prop_eq} guarantees that $x^*$ is at equilibrium
if and only if for each edge $(u,u')\in \activegraph(x^*)$:
    $\w x^*_u + \sum_{v \in \interact_u} x^*_v
        = \w x^*_{u'} + \sum_{v \in \interact_{u'}} x^*_v$.
In the global model, for any node $u$, the set $\interact_u$ consists
of all types but $u$ itself, implying that the sum cancels out, and we obtain
$(\w-1) x^*_u = (\w-1) x^*_{u'}$.
For $\w > 1$, this implies $x^*_u = x^*_{u'}$.
\end{proof}

We next show that the system always converges;
the proof relies on the key invariant that
for any $1 \leq k \leq n$, the total mass of the $k$
smallest types never increases over time.
More formally, we define the following quantities:

\begin{definition}
Let $y_1(t) \leq y_2(t) \leq \ldots \leq y_n(t)$
be the node masses sorted in non-decreasing order.
Define
\begin{align}
Y_k(t) & = \sum_{i \leq k} y_i(t)
       \; = \; \min_{R: |R| = k} \sum_{v \in R} x_v(t)
\label{eqn:partial-sum}
\end{align}
to be the sum of the masses of the $k$ smallest nodes at time $t$.
\end{definition}

The following lemma formally captures the notion that the rich get
richer in the global model.

\begin{lemma} \label{lem_glob_nondecreasing}
For every $k$, the function $Y_k(t)$ is non-increasing in $t$,
i.e., $\dot{Y}_k(t) \leq 0$.
\end{lemma}

\begin{proof}
Let $t,k$ be arbitrary.
Consider any set $S$ of $k$ nodes achieving the minimum in
\eqref{eqn:partial-sum} at time $t$; notice that there could be
multiple such sets $S$.
Consider any $u \in S, v \notin S$;
by definition of $S$, we have that $x_u(t) \leq x_v(t)$, and hence
$\flow{u}{v}(t) \geq 0$.
Because this holds for all such edges $(u,v)$, we obtain that the
total weight on nodes of $S$ cannot increase.
As this holds for all candidate sets $S$,
we get that $\dot{Y}_k(t) \leq 0$.
\end{proof}


\begin{theorem} \label{thm_glob_conv}
Under the global model, the system converges for any influence graph and any starting mass vector $x(0)$.
\end{theorem}

\begin{proof}
By Lemma \ref{lem_glob_nondecreasing}, each function $Y_j(t)$ is
non-increasing in $t$.
As all masses are non-negative, the $Y_j(t)$ are also bounded
below by $0$.
Hence, each function $Y_j(t)$ must converge to some value $Z_j$.
Thus, each function $y_j(t)$ must converge to $Z_j - Z_{j-1} =: z_j$.
It remains to show that this also implies convergence of $x(t)$.

Let $\delta > 0$ be at most the smallest difference between any two
distinct $z_j$, i.e., $\delta \leq \min_{i,j : z_i \neq z_j} |z_i-z_j|$.
Let $t_0$ be large enough that
$|y_{i}(t) - z_{i}| < \frac{\delta}{3}$ for all $i$ and $t \geq t_0$.

We will show that the only cases in which there could be
nodes $v$ and times $t' > t \geq t_0$ such
that $x_v(t) = y_j(t)$ and $x_v(t') = y_{j'}(t')$ is to have
$z_j=z_{j'}$.
If not, then let $\hat{t}$ be such that $|x_v(t) - z_j| < \delta/3$
for $t < \hat{t}$ arbitrarily close to $\hat{t}$,
and $|x_v(t) - z_{j'}| < \delta/3$
for $t' > \hat{t}$ arbitrarily close to $\hat{t}$.
Because $|z_{j'} - z_j| \geq \delta$, this implies that $x_v(t)$ must
be discontinuous at $t = \hat{t}$, which it cannot be.
\end{proof}

\subsection{Characterization of Lyapunov-Stable Equilibria}
For the global model, the properties required for Proposition
\ref{prop_gen_stb} hold for any independent set, since the interaction
mass of active types is always greater than the interaction
mass of inactive types.
Therefore, any equilibrium in which the set of active nodes is
independent is Lyapunov-stable.
To complete the characterization, we show that the converse is also
true.

\begin{theorem} \label{thm:global-stable}
In the global model with $\alpha > 1$,
an equilibrium $x^*$ is Lyapunov-stable if
and only if the active nodes form an independent set.
\end{theorem}

\begin{proof}
It remains to prove the ``only if'' direction. Assume that the active
nodes in an equilibrium $x^*$ do not form an independent set.
We will prove that $x^*$ is not Lyapunov-stable.

Let $C$ be a connected component of size $|C| \geq 2$
in $\activegraph(x^*)$. By the assumption that the active nodes in $x^*$
do not form an independent set, such a connected component exists.
Notice that each component of
$\activegraph(x^*)$ evolves in isolation, so we can focus on
only $C$ for the rest of the proof.
Therefore, by Corollary \ref{cor_equal-mass}, $x^*_v = \mu$
for all $v \in C$, for some value $\mu$.

Let $u,v \in C$ be two arbitrary nodes, and $\delta>0$ be
arbitrarily small.
Consider the following perturbation:
$x_u = x^*_u + \delta, x_v = x_v^* - \delta$, and $x_w = x^*_w$ for all
$w \neq u,v$.
By Theorem \ref{thm_glob_conv}, the system, starting from the perturbed
vector $x$, will converge to some new equilibrium $y$.
By Lemma \ref{lem_glob_nondecreasing}, the smallest mass of any
node in $C$ will always be at most $\mu - \delta$
during the process.
All $y$-active nodes must have the same mass;
therefore, if all nodes were active in $y$, they would
all have to have mass at most $\mu - \delta$,
which would imply that mass has disappeared from $C$, a
contradiction.
Hence, at least one node of $C$ must end up inactive in $y$.
In particular, this means that $||x^*-y||_1$ is not bounded
in terms of $\delta$, and $x^*$ is not Lyapunov-stable.
%
\end{proof}


\section{The Local Model} \label{sec-local}

In the previous section, we have given essentially complete
characterizations of convergence and stability of equilibria
under the global model, in which all types have the potential to
interact, even though only certain pairs of types can influence each
other (according to the graph $\infgraph$).

We now consider the local model, which is at the other extreme
of our general family: here, the interaction graph
$\intgraph$ is the same as the influence graph $\infgraph$;
hence, interactions occur only between individuals who also have the
potential to influence each other.
(We will generally denote this underlying graph by $\infgraph$,
with the understanding that $\intgraph = \infgraph$.)
We find that the problems of convergence and
stability are much more challenging in this case.
For the global model, we were able to
extract very useful organizing structures in the dynamical system
that gave us a natural progress measure toward convergence.
But as is well known, in general, a dynamical system on even
a small number of variables may have convergence properties
that are extremely difficult to analyze or express.
For example, not only does Lemma~\ref{lem_glob_nondecreasing}
  not hold for the node masses; a reformulation for \emph{interaction
    masses} does not hold, either.

Given the complex behavior of the update rules for the local model,
we find that the convergence and stability questions are already
difficult on graphs $\infgraph$ with a small number of nodes, and
we focus our results here on such cases.
(Of course, based on the motivating premise of the model, even
systems with a small number of variables are frequently natural,
corresponding to selection and influence dynamics in societies with,
for example, a small number of languages, a small number of
political parties, or a small number of dominant religions or cultures.)

We begin by considering the case $\w > 1$ and first prove the following theorem:
\begin{theorem}
Under the local model, if the influence graph is a $3$-path,
then the system converges from any starting state.
\label{thm-loc-3-w}
\end{theorem}

The full proof is provided in Appendix~\ref{app:3conv}, but we
provide a brief outline here.
The subtle difficulty arises due to the fact that the flow
between two types $u$ and $v$ does not necessarily go in the same
direction at all times, but instead may change its direction.
To keep track of the changes in direction, we define a
{\em configuration} of the system to be a labeling of all edges
$(u,v)$ in $\infgraph$ by the direction along which flow is traveling
(i.e., whether from $u$ to $v$ or from $v$ to $u$).
In the case of a 3-node path, there are four possible configurations.
We study transitions among the configurations as the system evolves
over time; we show that each configuration is either a {\em sink},
which cannot transition to any other configuration, or it has
the property that any change in the direction of an edge leads to
a sink configuration.  This ensures that there can be at most
one change in the direction of flow as the system evolves; hence, there
is a time $t_0$ such that for any $t>t_0$, no flow changes its
direction. After this point, Lemma~\ref{lem_gen_conv} guarantees
that the system converges.
For the case $\w\geq 2$, we show this fact only for the
3-path; for $\w<2$, we establish a more general result, showing the
same fact for arbitrary star graphs.

For $\w = 1$, we are able to prove convergence if the active subgraph
is a path of $n\leq 5$ nodes.
The proof requires different techniques than the ones we use for
$\w>1$: for paths of more than 3 nodes, flows on edges can
change their direction infinitely often.
The proof is provided in Appendix~\ref{app:5conv}. 


\subsection{Characterization of Universally Stable Equilibria}

Next, we turn our attention to Lyapunov-stable equilibria.
We focus on a very strong notion of stability: stability of
  an equilibrium $x$ simultaneously for all
$\w>1$.\footnote{%
Contrast this with the notion of stability used in the previous section ---
  there, we characterized Lyapunov-stable equilibria for any given fixed $\w$.}
Formally, we call a mass vector $x$ a
\emph{universally stable equilibrium} if $x$ is a Lyapunov-stable
equilibrium for every $\w > 1$.
Our goal here is to investigate which equilibria are universally
stable. Such equilibria are robust to (a very idealized notion of) a
change in the environment, as expressed by varying $\w$.

Our main result for universally stable equilibria is a complete
characterization for influence graphs that are forests, and
more generally for influence graphs whose connected components are
bipartite graphs.

\begin{theorem}\label{thm:local-univ-stability}
Assume that all connected components of the influence graph
$\infgraph$ are bipartite graphs.
Then, a mass vector $x^*$ is a universally stable equilibrium
under the local model if and
only if the distance between any two active nodes in $x^*$ is at least
$3$.
\end{theorem}

The proof of Theorem~\ref{thm:local-univ-stability} consists of
several sub-results, all of which hold for arbitrary influence graphs,
and imply the desired characterization under the assumption in the
theorem. It is worth noting that these sub-results constitute
significant progress towards understanding the structure of
universally stable equilibria for arbitrary influence graphs, as we
discuss later.
For brevity, if $x$ is a mass vector such that the distance between
any two active nodes is at least 3, we will say that $x$ is
\emph{3-separated}.

The first proposition proves the ``if'' direction of
Theorem~\ref{thm:local-univ-stability}.
Its proof follows from our analysis in Section~\ref{sec-general}.

\begin{proposition}\label{prop:3-separated}
Under the local model with $\w>1$, any 3-separated mass vector $x^*$
is a Lyapunov-stable equilibrium.
\end{proposition}

\begin{proof}
$x^*$ is an equilibrium by Proposition~\ref{prop_eq} since its active
nodes form an independent set.
By Proposition~\ref{prop_gen_stb}, an equilibrium whose active nodes
form an independent set is Lyapunov-stable if the interaction mass of each
active node is greater than the interaction mass of each of its neighbors.
For $\w>1$, this property holds when each inactive node has at most
one active neighbor.
In turn, this holds if and only if the distance between
every two active nodes is at least $3$;
thus, all such equilibria are Lyapunov-stable for every $\w > 1$.
\end{proof}

The ``only if'' direction of Theorem~\ref{thm:local-univ-stability} is
more complicated to prove. First, we show that if the active nodes of
an equilibrium do form an independent set, then being 3-separated is
necessary to ensure universal stability.

\begin{proposition} \label{prop:loc-dist3}
Let $x^*$ be a universally stable equilibrium,
and assume that the active nodes under $x^*$ form an independent set.
Then, $x^*$ is 3-separated.
\end{proposition}

\begin{proof}
For the sake of contradiction, suppose that $x^*$ is not 3-separated. Then there exists an inactive node $u$ with at least two active neighbors. Let $A_u$ be the set of all active neighbors of node $u$. We use the following notation:
$$ \textstyle
s=\sum_{v \in A_u} x^*_v; \quad
\eta = \min_{v \in A_u} x^*_v; \quad
\mu = \max_{v \in A_u} x^*_v.
$$
Define $\w = 1+ \eta^2$. We will show that $x^*$ is not Lyapunov-stable for this $\w$.

To prove instability, consider a perturbation $x$ which coincides with
$x^*$ on all nodes not in $\{u\} \cup A_u$, and satisfies
\begin{align}\label{eq:crazy-assumptions}
\begin{cases}
 x_v \leq x^*_v  \quad\text{for all}\; v \in A_u \\
 x_u =\delta \quad\text{for some}\;  \delta\in (0, s - \mu-\eta^2) \\
x_u + \sum_{v \in A_u} x_v = \sum_{v \in A_u} x^*_v = s.
\end{cases}
\end{align}
(Note that $s - \mu-\eta^2>0$ because $A_u$ consists of at least two nodes.)

Under such an $x$, the interaction mass of node $u$ is
\[
\wn_u \; = \; \w x_u + \sum_{v \in A_u} x_v
\; = \; \eta^2 \cdot x_u + s,
\]
while the interaction mass of any node $v\in A_u$ is
\[
\wn_v \; = \; \w x_v + x_u
\; \leq \; \w x^*_v + x_u
\; \leq \; (1+\eta^2)\mu + x_u.
\]

Since
    $x_u < s - \mu-\eta^2$,
we have that $\wn_v < (1+\eta^2)\mu + s - \mu-\eta^2 < s$,
and hence, $\wn_u > \wn_v$.
Thus, under this perturbation, mass starts flowing from all nodes
$v \in A_u$ to $u$, and this continues until $x_u \geq s - \mu-\eta^2$.
Consequently, the system cannot reach any equilibrium with
    $x_u < s - \mu-\eta^2$;
in particular, it cannot reach any equilibrium with
    $||x^*-x||_1 < s - \mu-\eta^2$.
Since this holds for arbitrarily small $\delta$,
and $s - \mu-\eta^2 > 0$ is a constant independent of $\delta$, we conclude that
$x^*$ is not Lyapunov-stable for this $\w$.
\end{proof}

With Proposition~\ref{prop:loc-dist3} in place, all that remains to
complete the proof of the ``only if'' direction of
Theorem~\ref{thm:local-univ-stability} is to ensure that the active
nodes in any universally stable equilibrium $x^*$ of a bipartite graph
form an independent set, i.e., that each connected
component $C$ of $\activegraph(x^*)$ consists of a single node.
This is implied by Lemma~\ref{lem:uniform}, which shows in general that
if $x^*$ is a universally stable equilibrium, then all the non-trivial connected
components of the subgraph of its active nodes are \emph{not} bipartite graphs.
This completes the proof of Theorem~\ref{thm:local-univ-stability}, as
any connected subgraph of a bipartite graph is a bipartite graph itself.


An additional benefit of Lemma~\ref{lem:uniform} is that it applies to
arbitrary influence graphs, and significantly limits the topologies a
connected component of $\activegraph(x^*)$ can have for a universally
stable equilibrium $x^*$. To state this lemma in the most general
form, we define a class of regular graphs which in particular
subsumes all bipartite graphs, all cliques, and all cycles whose
length is a multiple of $3$. We say that a $d$-regular graph is
\emph{locally balanced} if its
vertices can be partitioned into $k$ disjoint sets
$V_1, V_2, \ldots, V_k$ such that each vertex $v \in V_i$ has exactly
$d/(k-1)$ edges to each of the sets $V_j, j \neq i$.

\omt{
While we cannot prove this fact at present, we can significantly
restrict the remaining candidate subgraphs that \emph{can} occur as
connected components of $\activegraph(x^*)$, which we do with the next
two lemmas.
}

\begin{lemma} \label{lem:uniform}
Let $x^*$ be a universally stable equilibrium and
$C$ a non-trivial connected component of its active subgraph
$\activegraph(x^*)$. Then:
\begin{itemize}
\item[(a)] $C$ is a regular graph, and $x^*$ is uniform on $C$ (i.e., $x_u^*=x_v^*$ for all $u,v\in C$).
\item[(b)] $C$ is {\em not} a bipartite graph, and, more generally, $C$ is not locally balanced.
\end{itemize}
\end{lemma}

\begin{proof}
We begin by proving part (a).
Let $u,v \in C$ be a pair of \emph{adjacent} nodes.
The equilibrium conditions for $\w=2$ imply that
$ 2 x^*_v + \sum_{w \in \influence_v} x^*_w
= 2 x^*_u + \sum_{w \in \influence_u} x^*_w$,
and the ones for $\w=3$ that
$ 3 x^*_v + \sum_{w \in \influence_v} x^*_w
= 3 x^*_u + \sum_{w \in \influence_u} x^*_w$.
Subtracting the first equation from the second shows that
$x^*_v = x^*_u$.
Because $C$ is a connected component, applying this argument along all
edges in $C$ proves that all nodes in $C$ must have the
same mass $\mu$.

The interaction mass of node $v$ with $\w = 2$ is therefore
$\wn_v = \mu \cdot (|\influence_v \cap C| + 1)$.
Considering again a pair $u,v$ of adjacent nodes, the equilibrium
condition $\wn_u = \wn_v$ implies that
$|\influence_u \cap C| = |\influence_v \cap C|$.
Again by connectivity of $C$, this implies that all nodes
in $C$ have the same degree, so $C$ is regular.

Next, we prove part (b). Because $x^*$ is universally stable,
part (a) implies that $C$ is $d$-regular for some $d\geq 1$,
and $x^*_u = \mu$ (for some $\mu$) for all $u \in C$.
Assume for contradiction that $C$ is locally balanced, and
let $V_1, \ldots, V_k$ be the $k$ partitions of $C$.
Because\footnote{Recall that $\infgraph[S]$ denotes the
    induced subgraph of $\infgraph$ on the node set $S$.}
$\infgraph[V_i \cup V_j]$ is a $d/(k-1)$-regular bipartite
graph for each pair $i \neq j$, all partitions $V_i$ must have the
same size $s = |C|/k$.

Set $\w=d+1$, and let $\delta>0$ be arbitrary.
Consider perturbed vectors of the following form:
$x_v = x^*_v+\frac{1}{s} \cdot \delta$ for every $v \in V_1$ and
$x_u=x_u-\frac{1}{s(k-1)}\cdot \delta$ for every $u \notin V_1$.
(That is, a total mass of $\delta$ is removed uniformly from
nodes not in $V_1$, and added uniformly over the nodes in $V_1$.)

In moving from $x^*$ to $x$, the interaction mass of each node
$v \in V_1$ changes by
$\w \cdot \frac{1}{s} \cdot \delta
  - d \cdot \frac{1}{s(k-1)} \cdot \delta > 0$,
while the interaction mass of each node $u \notin V_1$ changes by
\begin{multline*}
-\w \cdot \frac{1}{s(k-1)} \cdot \delta
- \frac{d(k-2)}{k-1} \cdot \frac{1}{s(k-1)} \cdot \delta
+ \frac{d}{k-1} \cdot \frac{1}{s} \cdot \delta \\
=
\Big(-(d+1)-\frac{d(k-2)}{k-1}+d\Big) \cdot \frac{1}{s(k-1)} \cdot \delta
\; < \; 0.
\end{multline*}
Thus, for any such vector $x(t)=x$, all flows are directed from
nodes not in $V_1$ to nodes in $V_1$.
Furthermore, by symmetry of the original vector $x^*$ and the
perturbation, the mass vectors $x(t')$ for $t' > t$ will be
of the same form, for a different $\delta' > \delta$.
Thus, the same argument will apply at all times.
Hence, the direction of flows never changes,
and Lemma~\ref{lem_gen_conv} guarantees that the system converges.
Since the interaction mass of all nodes in $V_1$ is only increasing,
and the interaction mass of all nodes not in $V_1$ is only decreasing,
the only equilibrium $y$ the system can converge to is one
in which all nodes outside of $V_1$ have zero mass.
In particular, this means that even starting from
$||x^*-x||_1 = \delta'$ (which would correspond to using
$\delta=\delta'/2$ in our analysis),
$||x^*-y||_1$ is not bounded in terms of $\delta'$,
so $x^*$ is not Lyapunov-stable.
\end{proof}


Lemma \ref {lem:uniform} considerably narrows down the set of
equilibria for which the question of whether or not they are
universally stable remains open.

More specifically, it only remains to
consider mass vectors $x$ in which there is a non-trivial
connected $C$ of $\activegraph(x)$ such that
$C$ is a $d$-regular graph (for some
$d \geq 1$), is not a locally balanced graph (in particular not a bipartite
graph), and for every $u,v \in C$, $x_v=x_u$.
We conjecture that such mass vectors are not universally stable;
it would then follow that in any universally stable equilibrium,
all components have size $1$, and hence by Proposition~\ref{prop:loc-dist3}
the active nodes would be at mutual distance $3$.
Accordingly, we formulate the following:

\begin{conjecture} \label{conj:stability}
Under the local model, a mass vector is a universally stable equilibrium
if and only if its active nodes are at pairwise distance
at least $3$ in the influence graph.
\end{conjecture}


\newcommand{\half}{\ensuremath{\frac{1}{2}}\xspace}

\subsection{Characterization of Lyapunov Stable Equilibria for $\w=1$}
For $\w>1$, we have shown that any equilibrium whose active nodes
form an independent set of pairwise node distance at least 3
is Lyapunov-stable. Perhaps surprisingly, this ceases to be true for $\w=1$.
Indeed, on the path of length $4$, the equilibrium $x^*=(\half,0,0,\half)$
is not Lyapunov-stable.

We can see this instability as follows. Consider vectors of the form
$x^{(\delta)} = (\half-\delta,\delta,\delta,\half-\delta)$.
Under $x^{(\delta)}$, for any $\delta \in (0,\half)$, the interaction mass
of nodes 1 and 4 is strictly smaller than the interaction mass of
nodes 2 and 3 (whose interaction masses are equal).
This implies that no vector $x^{(\delta)}$ can be an equilibrium for
$\delta \in (0,\half)$, and that flow will always be directed from
nodes 1 and 4 to nodes 2 and 3. Furthermore, the flow from node 1 to
node 2 is equal to the flow from node 4 to node 3, implying that at
later times, the mass vector will be of the form $x^{(\delta')}$
with $\delta' > \delta$. As we know by Theorem \ref{thm_loc_5}
that the 4-path always converges, the system converges to some mass vector
    $y=x^{(\delta^*)}$
such that $\delta^*>0$.\footnote{%
This also follows directly from our argument with $x^{(\delta)}$.}
Since the update rule is continuous, this $y$ must be an equilibrium. We have proved that the only such equilibrium is the one with $\delta^*=\tfrac12$. Thus, starting from the perturbation
$x^{(\delta)}$ of $x^*$, the system can only converge to a state
$y$ in
which $y_1 = y_4 = 0$.


While a pairwise distance of 3 between active nodes is not enough
to guarantee stability, a pairwise distance of 4 is sufficient.

\begin{theorem} \label{thm_loc_ind}
Let $x^*$ be a mass vector whose active nodes have pairwise
distance at least 4.
Then, $x^*$ is a Lyapunov-stable equilibrium for $\w=1$.
\end{theorem}

\begin{proof}
The proof is much more involved than the proof of
Proposition~\ref{prop_gen_stb}, for the following reason:
even for arbitrarily small perturbations to $x^*$, it is
possible that inactive neighbors $v$ of an active node $u$ have
higher interaction mass; thus, the conditions of
Proposition~\ref{prop_gen_stb} do not apply, and in fact,
$u$ could lose mass over time.
However, we will be able to show that the total mass $u$ loses,
starting from a perturbation of magnitude at most $\delta$,
is bounded by a function $g(\delta) \to 0$ as $\delta \to 0$.

Let $A$ be the set of all $x^*$-active nodes, and
let $x(0)$ be a perturbation of $x^*$ with
$||x^* - x(0)||_1 \leq \delta \leq  \tfrac18 \cdot \min_{u \in A} x_u^*$.
We will show below that for each node $u \in A$, and all times $t$,
we have that $|x_u(t) - x^*_u| \leq 2\delta$.
Because
\[
\sum_{v \notin A} | x_v(t) - x_v^* |
\; = \; \sum_{v \notin A} x_v(t)
\; = \; \sum_{u \in A} (x_u^* - x_u(t))
\; \leq \; \sum_{u \in A} |x_u(t) - x_u^*|,
\]
we obtain that
\[
||x(t) - x^*||_1
\; = \; \sum_{u \in V} |x_u(t) - x_u^*|
\; \leq \; 2 \sum_{u \in A} |x_u(t) - x_u^*|
\; \leq \; 4n\delta
\; \to \; 0 \;\; \mbox{ as } \delta \to 0.
\]

It remains to prove the inequality $|x_u(t) - x_u^*| \leq 2\delta$
for all nodes $u \in A$ and times $t$.
Define $W = V \setminus (A \cup \bigcup_{u \in A} \influence_u)$ to be
the set of all nodes at distance at least 2 from all active nodes.
We will prove the inequality by showing that any
\netflow from $u$ to its neighbors $v$ can be ``charged'' against
\netflow from $W$ to $v$.
More formally, we will simultaneously prove the following invariants
for all times $t$ and any set $U \subseteq A$:
\begin{subequations} \begin{align}
\sum_{w \in W} x_w(t) & \leq \sum_{w \in W} x_w(0),
\label{eqn:upper-bound-xw}\\
\sum_{u \in U} x_u(t) & \geq \sum_{u \in U} x_u(0) - \sum_{w \in W} (x_w(0) - x_w(t))
\quad \text{ for any set } U \subseteq A.
\label{eqn:lower-bound-xu}
\end{align} \end{subequations}

Let $t$ be an arbitrary time and assume that the invariants
hold at time $t$.
First, we notice some useful consequences of the invariants, including
the desired fact that $|x_u(t) - x_u^*| \leq 2\delta$.

From Inequality~\eqref{eqn:upper-bound-xw}, we get that
$\sum_{w \in W} x_w(t) \leq \sum_{w \in W} x_w(0) \leq \sum_{v \notin A} x_v(0) \leq \delta$.
Substituting this bound into Inequality \eqref{eqn:lower-bound-xu} with
$U = \{u\}$, and using that $x_u(0) \geq x^*_u - \delta$ gives us that
$x_u(t) \geq x^*_u - 2\delta$.
Similarly, using Inequality~\eqref{eqn:lower-bound-xu} with
$U = A \setminus \{u\}$ gives us an upper bound of
$x_u(t) \leq x^*_u + 2\delta$.
So we have shown that $|x_u(t) - x_u^*| \leq 2 \delta$.

Let $(w,v), w \in W, v \notin W$ be an arbitrary edge,
and $u$ the unique active neighbor of $v$ in $\infgraph$.
The \netflow on the edge $(w,v)$ is
$\flow{w}{v}(t) ={\ip \cdot} \frac{x_w(t)\, x_v(t) (\wn_v(t)-\wn_w(t))}{\wn_w(t)\, \wn_v(t)}$.
We have just seen that $x^*_u - 2\delta \leq x_u(t) \leq x^*_u + 2\delta$,
so we can also bound $\wn_v(t) \geq x^*_u - 2\delta$.
Applying Inequality~\eqref{eqn:lower-bound-xu} with
$U = A \setminus \{u\}$ also gives us an upper bound of
$\wn_v(t) \leq 1-\sum_{u' \in A, u' \neq u} x_{u'}(t) \leq x^*_u + 2\delta$.
Furthermore, using the definition of $W$ and
Inequality~\eqref{eqn:lower-bound-xu} for $U=A$,
\[
\wn_w(t)
\; \leq \; \sum_{v \notin A} x_v(t)
\; \leq \; \sum_{v \notin A} x_v(0) + \sum_{w \in W} (x_w(0) - x_w(t))
\; \leq \; 2\delta.
\]
Substituting these bounds, we get that
\begin{align}
\flow{w}{v}(t) & \geq
\ip \cdot \frac{x_w(t)\, x_v(t) (x^*_u-4\delta)}{2\delta (x^*_u+2\delta)}.
\label{eqn:upper-bound-wv-flow}
\end{align}

By definition of $\delta$, this quantity is always non-negative.
In particular, this means that \netflow goes from $w$ to $v$; since the
edge $(w,v)$ was arbitrary, we have established the
invariant~\eqref{eqn:upper-bound-xw}.

Next, fix an arbitrary node pair $u \in A, v \in \influence_u$.
The \netflow from $u$ to $v$ is
\begin{align*}
\flow{u}{v}(t)
& = \ip \cdot \frac{x_u(t)\, x_v(t) (\wn_v(t)-\wn_u(t))}{\wn_u(t)\, \wn_v(t)}
\; \leq \;
\ip \cdot \frac{x_u(t)\, x_v(t) \sum_{w \in W \cap \influence_v} x_w(t)}{(x_u(t))^2} \\
& = {\ip \cdot} \frac{1}{x_u(t)} \sum_{w \in W \cap \influence_v} x_w(t)\, x_v(t)
 \; \leq \;
    \ip \cdot \frac{1}{x_u^* - 2\delta} \sum_{w \in W \cap \influence_v} x_w(t)\, x_v(t).
\end{align*}
On the other hand, summing the bound \eqref{eqn:upper-bound-wv-flow}
over all nodes $w \in W \cap \influence_v$, we get that
\begin{align*}
\sum_{w \in W \cap \influence_v} \flow{w}{v}(t)
& \geq
\ip \cdot \frac{x^*_u-4\delta}{2\delta (x^*_u+2\delta)} \cdot
\sum_{w \in W \cap \influence_v} {x_w(t)\, x_v(t)}.
\end{align*}

Because $\delta \leq x^*_u/8$, we get that
$\frac{1}{x_u^* - 2\delta}
\leq \frac{x^*_u-4\delta}{2\delta (x^*_u+2\delta)}$,
so the \netflow from $u$ to $v$ is at most the total \netflow from
all $w \in W \cap \influence_v$ to $v$.
For any set $U \subseteq A$, summing this inequality over all
$u \in U$ (and noticing that we never double-count the same edge)
now shows that the total flow out of $U$ is no more than the total
  flow out of $W$; hence, the decrease in $U$'s mass can be charged to
  a corresponding decrease of mass in $W$, and we have
established Invariant~\eqref{eqn:lower-bound-xu}. 
\end{proof}


\omt{
We begin
by providing an example in which an independent set with distance $3$
between nodes is not Lyapunov-stable:
\begin{claim}
For a path of length $4$, the equilibrium $x_1=0.5,x_2=0,x_3=0,x_4=0.5$ is not Lyapunov-stable.
\end{claim}
\begin{proof}
Pick $\delta>0$ arbitrarily small.
Let $\delta' = \frac{\delta}{4}$.
Consider the perturbation $x'_1=0.5 - \delta'$,
$x'_2=\delta'$, $x'_3=\delta'$, $x'_4=0.5-\delta'$.
After the perturbation is performed, $x_2$ and $x_3$ have greater interaction
mass than $x_1$ and $x_4$.
Since $x_2=x_3$ by symmetry we have that
flow only goes from $1$ to $2$ and from $4$ to $3$ and does not change
its direction. Hence, the perturbed system converges. It has to
converge to an equilibrium in which $x'_1$ and $x'_4$ are $0$, since
this is the only possible equilibrium once we have that $x'_2$ and
$x'_3$ are always positive.
\end{proof}}

\omt{
Loosely speaking,
this is done by bounding the mass that an $x^*$-active node can lose to
its neighbor $v$ by the mass that $v$ can gain from nodes $w$ which
are not neighbors of any $x^*$-active nodes. The proof is completed by
noting that the last quantity is bounded by the mass of nodes $w$
which are not neighbors of any $x^*$ active nodes which cannot gain
any extra mass.

\begin{proof}
Let $A_i$ be the set of all nodes at distance $i$ from active nodes
(for example $A_0$ is the set of all active nodes in $x^*$). In Claim
\ref{clm_loc_flows} we show that there exists a $\delta$ arbitrarily
small such that for every perturbation of size $\delta$ all
$x^*$-active can only lose a small amount of mass. Or more formally,
for every time $t$: $\forall u \in A_0$,  ${x_u^*} - 2n\delta < x_u(t)
\leq {x_u^*} + \delta$, where $n=|A_0|$ . Note that this property
immediately implies stability as we have that for all time $t$:
$\sum_{u \in A_0} |x_u^* - x_u(t) | < 2n^2 \delta$, $\sum_{v \in A_1}
|x_v^* - x_v(t) | < 2n^2 \delta$ and $\sum_{w \in A_{\geq 2}} |x_w^* -
x_w(t) | <  \delta$. Therefore we have that for every time $t$: $||x^*
- x(t)||_1 < (4n^2+1) \delta$. This in turn implies that the $x^*$ is
Lyapunov-stable.
\end{proof}} 



\section{Discussion and Conclusions}
\label{sec:conclusions}

In this paper, we presented a novel model of cultural dynamics that captures the
essential aspect of several previously studied models:
the interplay between selection and influence.
We concentrated on two instances of this model. In the basic version,
the {\em global model}, each person selects another person from the entire
population to interact with. In the {\em local model}, a person selects an
interaction partner from a subset of the population consisting of
similar people.
We provided a nearly complete treatment of the global model, showing
that the system always converges from any initial mass vector,
and providing a complete characterization of Lyapunov-stable equilibria.

\subsection{Modeling Choices}
\label{sec:modeling-choices}

\subsubsection{Continuum of Individuals}
We assumed a continuum of individuals, rather than a finite population.
With finite populations, convergence to equilibrium states
is quite immediate. The directed Markov Chain of all possible
assignments of individuals to nodes of the graph is finite. For any
state in which the occupied nodes do \emph{not} form an independent
set, there is a sequence of finitely many moves (which has strictly
positive probability of occurring) which will result in the
individuals being located at an independent set; the latter states are
sinks of the Markov Chain. Thus, convergence in finite time is always
guaranteed with probability 1. The primary focus of studying such
finite Markov Chains would be a focus on the \emph{amount of time} it
would take to reach a sink state.

For very large populations, the predictions of the convergence time
may be of less interest, and we believe that a focus on the stability
of equilibria, and the convergence guarantees in the limit of large
populations, are of interest in understanding the outcomes of the
influence-selection process.

While several articles such as the early work of
Kurtz~\cite{kurtz1970solutions} and the more refined analysis of
Wormald~\cite{wormald:differential-survey} establish precise
connections between discrete-time processes with finite populations
and the mean-field continuous-time limit as both time and the
population are scaled, we do not believe that these approaches are
sufficiently powerful to easily imply results presented here; they
may, however, establish that with high probability, the discrete
version of the problem stays close to the mean-field approximation.
These results also motiviate the continuous-time version of the
problem studied here.

\subsubsection{Continuous Time}
In keeping with much of the literature on population dynamics, we
treat time as continuous.
However, all results proved here hold equally for discrete time;
indeed, an earlier version of this article --- still available on the
arXiv at {http://arxiv.org/abs/1304.7468} (v.1) ---
carried out all proofs in discrete time.
In discrete time, the flows defined in Equation~\eqref{eq:model-flows}
do not correspond to continuous derivatives of the node masses, but
rather to discrete changes from step $t$ to $t+1$.
The proofs of all results stay essentially the same, requiring only
the obvious modifications.
The main exceptions are Lemma~\ref{lem_glob_nondecreasing} and
Theorem~\ref{thm_glob_conv}, whose proofs become slightly more
intricate: a simple continuity argument cannot be applied, as discrete
changes may lead to jumps in the $x_u(t)$. However, careful accounting
and judicious choices of interval sizes $\delta$ still make the proofs
go through.

One argument frequently raised against discrete-time analysis is that
it may lead to oscillations in states, in particular when states are
updated synchronously.
In our proofs, we have not observed such oscillations, and indeed
conjecture that both the discrete-time and continuous-time versions
will always converge to an equilibrium.
In the cases where a proof of this convergence has been difficult,
this difficulty has persisted in both continuous and discrete time.


%

\subsubsection{Discrete Graph Structures}
Throughout, we have assumed that the interaction and influence graphs
are ``discrete,'' in the sense that the propensity of individuals to
interact with (or be influenced by) individuals of adjacent types is
the same for all adjacent types.
In reality, the world will not be as black-and-white; rather, there
will be some types $v$ adjacent to $u$ that are more likely than
others to succeed in convincing individuals from node $u$ to switch.
For instance, in the introductory example from
Section~\ref{sec-intro}, while radical protestants may be most likely
to become moderate protestants (or stay radical), a small fraction may
directly become atheists.

In a more general form of the model, the probability $\ip$ for
switching between types would depend on the specific types, i.e., be
of the form $\ip_{u,v}$, where it is expressly possible that
$\ip_{u,v} \neq \ip_{v,u}$. $\ip_{u,v} = 0$ would correspond to the
absence of an edge from $u$ to $v$.
Similarly, we could assign weights to the edges of the interaction
graph, and have meeting probabilities follow those weights.

This more general model becomes significantly more complex to
analyze.
When $\ip_{u,v} = \ip_{v,u}$ for all pairs $(u,v)$, much of the
analysis in the present work carries over, but for asymmetric
versions, a more complex approach may be needed.
Similarly, given the difficulties caused even by the simple local
model, a more general weighted interaction graph model looks like a
rather formidable challenge.

\subsection{Open Questions}
An open question is to predict the equilibrium to which the system
converges starting from a given initial mass vector.
We suspect that with probability 1 over possible starting
states, the system converges to an equilibrium in which the active
nodes form an independent set.

The local model involves, at its heart, a dynamical system on the
population fractions that is complicated even for small numbers
of variables.  As such, it raises many interesting and challenging
questions, and we have made progress on some of these.
In particular, we know that on paths of length 3
(for $\alpha > 1$) and at most 5 (for $\alpha = 1$),
the system converges from any starting state.
However, it is open whether convergence occurs for all graphs.
On the stability frontier, for $\alpha>1$, we
conjecture that the only Lyapunov-stable equilibria are those in which the
active nodes have pairwise distance at least $3$.
We showed that such equilibria are
indeed Lyapunov-stable, and that a number of other equilibria are not 
Lyapunov-stable --- including ones in which the active nodes form any
other independent set, or ones in which they form a
{\em locally balanced graph} (a class that includes bipartite graphs).
Finally, we would like to raise an even
more challenging question: does the dynamical system defined
by the general model --- or even the general model --- always converge?

\bibliographystyle{plain}
\bibliography{header,refs,new-refs}

\newpage{}
\appendix
\section{ Convergence on a 3-node path (Proof of Theorem 4.1)} \label{app:3conv}

Observe that the flow between two types $u$ and $v$ does not
necessarily go in the same direction at all times, but
instead may change its direction.
To keep track of the changes in direction, we define a
{\em configuration} of the system to be a labeling of all edges
$(u,v)$ in $\infgraph$ by the direction along which flow is traveling
(i.e., whether it travels from $u$ to $v$ or from $v$ to $u$).
A configuration which cannot transition to any other configuration is
called a \emph{sink configuration}. Sink configurations are important
because they guarantee convergence by Lemma \ref{lem_gen_conv}.

In the case of a 3-node path and $\w\geq 2$, there are four possible
configurations.
We study transitions among the configurations as the system evolves
over time; we show that each configuration is either a sink
configuration, or it has the property that any change in the direction
of an edge leads to a sink configuration.
This ensures that there can be at most one change in the direction of
flow as the system evolves.

For a 3-node path and $\w<2$, convergence will follow as a special
case of the more general Theorem~\ref{thm:local-star}, which
establishes convergence for all star graphs when $\w < 2$. For a 3-node path and
$\w \geq 2$, we prove the following lemma (which, jointly with Lemma~\ref{lem_gen_conv},
implies Theorem~\ref{thm-loc-3-w} for $\w\geq 2$).


\begin{lemma} \label{lem_loc_3_w}
Consider the local model with $\w \geq 2$ such that the influence graph
$\infgraph$ is a 3-node path. Then, there is a time $t_0$ such that for any $t>t_0$, no flow changes its
direction.
\end{lemma}

\begin{proof}
Let the nodes of the path be $(1,2,3)$, in order.
Consider an arbitrary time $t$,
and recall that $\wn_1(t) = \w x_1(t)+x_2(t)$.
Node 1's interaction mass is decreased at a rate of
$\alpha \flow{1}{2}(t)$ from flow leaving node $1$ to node $2$,
and increased at a rate of $\flow{1}{2}(t) + \flow{3}{2}(t)$
from flow entering node $2$.
By applying the same reasoning to nodes $2$ and $3$, we get:
\begin{equation} \label{eqn:recursive} \begin{split}
\dot{\wn}_1(t) & = \flow{3}{2}(t) - (\w-1) \flow{1}{2}(t), \\
\dot{\wn}_2(t) & = (\w-1) (\flow{1}{2}(t) + \flow{3}{2}(t)),\\
\dot{\wn}_3(t) & = \flow{1}{2}(t) - (\w-1) \flow{3}{2}(t).
\end{split} \end{equation}
Let $x_i = x_i(t), \wn_i = \wn_i(t), \flow{i}{j} = \flow{i}{j}(t)$
for $i,j = 1, 2, 3$.
We will distinguish three cases based on the relative
sizes of $\wn_1, \wn_2, \wn_3$.

\begin{enumerate}
\item If $\wn_2 \geq \wn_1$ and $\wn_2 \geq \wn_3$, then
both $\flow{1}{2}$ and $\flow{3}{2}$ are non-negative.
According to Equation~\eqref{eqn:recursive},
$\wn_2$ increases by at least as much as both
$\wn_1$ and $\wn_3$, so the same inequality will subsequently
  as well.
Thus, we have reached a sink configuration.

\item If $\wn_2 < \wn_1$ and $\wn_2 < \wn_3$, then
both $\flow{1}{2}$ and $\flow{3}{2}$ are negative.
By Equation~\eqref{eqn:recursive},
$\wn_2$ decreases by at least as much as both
$\wn_1$ and $\wn_3$, so again, the inequalities will hold forever, and
we have reached a sink configuration.

\item The remaining case is that $\wn_2 < \wn_1$ and $\wn_2 \geq \wn_3$.
(The case $\wn_2 < \wn_3, \wn_2 \geq \wn_1$ is symmetric.)
Here, $\wn_3$ decreases, $\wn_1$ increases, and $\wn_2$ may increase
or decrease. If the relative order of $\wn_1, \wn_2, \wn_3$ stays the
same for all times after $t$, then we have reached a sink configuration.
Otherwise, at some time $t' \geq t$, we must reach either a
configuration with $\wn_2(t') < \wn_1(t'), \wn_2(t') < \wn_3(t')$ or with
$\wn_2(t') \geq \wn_1(t'), \wn_2(t') \geq \wn_3(t')$.
Either of those configurations is a sink configuration by the preceding two cases.
\end{enumerate}
\omt{
\item $\w < 2$.
\begin{enumerate}
\item If $\wn_2 \geq \wn_1$ and $\wn_2 \geq \wn_3$, then
$\wn_2$ increases, but it is possible that one of
$\wn_1, \wn_3$ increases strictly more than $\wn_2$.
If this does not happen for any time $t' \geq t$, then
we again have a sink configuration.
Otherwise, let $t'$ be the first time such that
$\wn_2(t') < \wn_1(t')$ or $\wn_2(t') < \wn_3(t')$.
We will show in the next case that it is impossible for both
inequalities to hold simultaneously.
Thus, without loss of generality, $\wn_2(t') < \wn_1(t')$ and
$\wn_2(t') \geq \wn_3(t')$. This will be analyzed as the third case.

\item It is impossible that $\wn_2 < \wn_1, \wn_2 < \wn_3$:
substituting the definitions of $\wn_i$ shows that this would imply
$x_3 < (\w-1)(x_1-x_2)$ and
$x_1 < (\w-1)(x_3-x_2)$, which is a contradiction for $\w < 2$.

\item If $\wn_2 < \wn_1, \wn_2 \geq \wn_3$ (or the symmetric case),
then $\flow{1}{2} \leq 0, \flow{3}{2} \geq 0$.
Because $\w < 2$, by Equation~\eqref{eqn:recursive},
the increase in $\wn_1$ is at least as much as any
increase in $\wn_2$, and the decrease in $\wn_3$ is at least as much
as any decrease in $\wn_2$. Thus, the ordering is preserved forever,
and this configuration is a sink configuration.
\end{enumerate}
}
In summary, each configuration is either a sink configuration, or will reach a sink
configuration at the next transition to a different order of interaction
masses.
\end{proof}

Next we prove that for $\w<2$ the process converges on every star graph
(and in particular on the 3-path).

\begin{theorem}\label{thm:local-star}
Under the local model with $\w<2$, if the influence graph is a star graph,
then the system converges from any starting state.
\end{theorem}

In the remainder of this section, we prove
Theorem~\ref{thm:local-star}.
More specifically, we show that eventually the system enters a sink
configuration.
The first lemma towards the proof holds for arbitrary values of $\w$.

\begin{lemma}\label{lm:local-star-UB}
Consider the local model with an arbitrary $\w \geq 1$
such that the influence graph $\infgraph$ is a star graph.
Then, at any time, the number of edges with flow directed
away from the center is at most $\lfloor \w \rfloor$.
\end{lemma}
\begin{proof}
Denote the central node by $u$.
Fix some time $t$, and let $R$ be the set of all peripheral nodes $v$
such that the flow on the edge $(u,v)$ is directed from $u$ to $v$.
Because the flow is directed towards $v$,
$\wn_v > \wn_u$ for all $v \in R$.
Rearranging this inequality gives us that
$(\w-1)(x_v-x_u) > \sum_{w\neq u,v} x_w$, which implies in particular that
$x_v > \frac{\sum_{w \neq u,v} x_w}{\w-1}$.
Summing over all $v \in R$ now implies that
\begin{align*}
\sum_{v \in R} x_v
& > \sum_{v \in R} \frac{\sum_{w \neq u,v} x_w}{\w-1}
\; \geq \; (|R|-1) \frac{\sum_{v \in R} x_v}{\w-1}.
\end{align*}
Thus we have that $|R| \leq \lfloor \w \rfloor$.
\end{proof}

\begin{lemma}\label{lm:local-star-w}
Consider the local model such that the influence graph
$\infgraph$ is a star graph, and with $\w<2$.
Then, any configuration in which flow on exactly one edge is directed
away from the center node is a sink configuration.
\end{lemma}
\begin{proof}
Let $u$ be the center node.
Suppose that at time $t$, the system is in a configuration in which
the flow on exactly one edge $(u,v)$ is directed away from
the center; so $\wn_u(t) < \wn_v(t)$.
By Lemma~\ref{lm:local-star-UB} there can be at
most one edge on which the flow is directed away from $u$,
and $(u,v)$ is such an edge.
The changes in interaction masses are
\begin{align*}
\dot{\wn}_u(t) &= (\w-1) \sum_{w \neq u,v}\flow{w}{u}(t) - (\w-1)\flow{u}{v}(t) \\
\dot{\wn}_v(t) &= (\w-1)\flow{u}{v}(t) + \sum_{w \neq u,v}\flow{w}{u}(t).
\end{align*}
Their difference is
\begin{align*}
\dot{\wn}_u(t)- \dot{\wn}_v(t)
& \leq (\w-2) \sum_{w \neq u,v}\flow{w}{u}(t)
    -2(\w-1)\flow{u}{v}(t).
\end{align*}
Because $\w < 2$, the right-hand side is negative, so $u$'s
interaction mass grows more slowly (or decreases faster) than $v$'s,
implying that the edge $(u,v)$ remains directed from $u$ to $v$.
Hence, the configuration is a sink configuration.
\end{proof}

Theorem~\ref{thm:local-star} now follows from
Lemmas~\ref{lem_gen_conv}, \ref{lm:local-star-UB} and
\ref{lm:local-star-w}, as follows.
If the system ever enters a configuration in which exactly one edge
has flow directed away from the center, then by
Lemma~\ref{lm:local-star-w}, it subsequently stays in this
configuration forever, so by Lemma~\ref{lem_gen_conv},
the system converges.
By Lemma~\ref{lm:local-star-UB}, the only other alternative is that the
system is always in the configuration with all edges directed inwards;
then, again, it converges by Lemma~\ref{lem_gen_conv}.

\section {Local model with $\w=1$: Convergence on a Path}
\label{app:5conv}
Assume that the active subgraph is an $n$-node path
with nodes $(1, 2, \ldots, n)$.
The endpoints of the path, nodes $1$ and $n$, always have interaction
masses no larger than their neighbors (nodes $2, n-1$), implying that
their masses $x_1(t), x_n(t)$ are monotonically non-increasing.
This implies convergence of $x_1(t)$, $x_n(t)$ as $t \to \infty$.
In the following proposition, we will exploit the convergence at
the endpoints to show that $x_2(t)$ and $x_{n-1}(t)$ must also
converge.
For a path of length at most $5$, this implies convergence of the
vector $x$ to an equilibrium, as the total mass stays constant.
Our technique does not apply beyond length $5$;
we do not know of a direct way to generalize
the argument inductively to paths of arbitrary lengths.


\begin{theorem} \label{thm_loc_5}
Consider the local model with $\w=1$.
If the influence graph is a path of $n \leq 5$ nodes,
then the system converges.
\end{theorem}

\begin{proof}
We already argued above that $x_1(t)$ and $x_n(t)$ converge.
If the path has 3 nodes, then $x_2(t)=1-x_1(t)-x_3(t)$ (by mass
conservation), so $x_2(t)$ converges as well.
So assume that $n \in \{4,5\}$.
Below, we show that $x_2(t)$ converges as well;
a symmetric argument applies to $x_{n-1}(t)$.
If the path has $4$ nodes, we are done at this point.
If the path has $5$ nodes, then
$x_3(t) = 1-x_1(t)-x_2(t)-x_4(t)-x_5(t)$ must converge as well.
Thus, $x(t)$ converges in all cases.

To prove that $x_2(t)$ converges, we distinguish two cases, based
on $y_1 = \lim_{t \to \infty} x_1(t)$.
\begin{enumerate}
\item If $y_1 = 0$, there are two subcases.
If $x_1(t) \geq x_2(t)$ for all $t$, then clearly, $x_2(t) \to 0$ as well.
Otherwise, there exists a $t_0$ with $x_2(t_0) > x_1(t_0)$.
By Equation~\eqref{eq:def-by-flows}, specialized to the local model and $\w=1$,
we obtain that for any $t$,

\begin{align*}
\dot{x}_1(t) &=  \ip \cdot x_1(t) \cdot
          \left( \frac{x_1(t)}{\wn_1(t)} + \frac{x_2(t)}{\wn_2(t)} - 1 \right),\\
\dot{x}_2(t) &= \ip \cdot  x_2(t) \cdot
          \left( \frac{x_1(t)}{\wn_1(t)} + \frac{x_2(t)}{\wn_2(t)}
            + \frac{x_3(t)}{\wn_3(t)} - 1\right).
\end{align*}

Then, clearly, $x_1(t) < x_2(t)$ implies $\dot{x}_1(t) < \dot{x}_2(t)$.
In particular, this means that $x_2(t) > x_1(t)$ for all $t \geq t_0$.
In turn, this inequality is used in the last step of the following
derivation:
\begin{align*}
\max(\flow{3}{2}(t),0)
& \leq {\ip \cdot}\frac{x_2(t)\, x_3(t)}{\wn_2(t)\, \wn_3(t)} \cdot x_1(t)\\
& = {\ip \cdot}\frac{x_1(t)\, x_2(t)\, x_3(t)}{\wn_1(t)\, \wn_2(t)} \cdot \frac{\wn_1(t)}{\wn_3(t)}\\
& =  \flow{1}{2}(t) \cdot \frac{x_1(t)+x_2(t)}{x_2(t)+x_3(t)+x_4(t)}\\
& \leq 2\, \flow{1}{2}(t).
\end{align*}
Thus, the total amount of flow entering node 2 after time $t_0$ is at
most $3 \int_{t_0}^{\infty} \flow{1}{2}(t) dt \leq 3\, x_1(t_0)$.
The reason for the last inequality is that flow never enters node 1, so
the total amount of flow that can leave node 1 for node 2 after $t_0$
is at most the amount that was at node 1 at time $t_0$.

Let $F^+(t)$ (resp., $F^-(t)$) be the total amount of flow
that has entered (resp., left) node 2 up to time $t$.
We have just proved that $F^+(t) - F^+(t_0) \leq 3 x_1(t_0)$.
Therefore, ${F^+(t)}$, being monotone and bounded, must converge.
Because flow can only leave node 2 when it was already there, we get
that $F^-(t) \leq x_2(0) + F^+(t)$ is also bounded, and must also converge.
Hence, $x_2(t) = x_2(0) + F^+(t) - F^-(t)$, being the difference
between two convergent functions, must also converge.

\item If $y_1 > 0$, we will pursue a similar argument, but this time
bounding the cumulative flow \emph{out of} node 2 instead of into it.
Because $x_2(t) \leq 1$ for all times $t$, this means that the
cumulative flow into node $2$ must also be bounded. Then, an identical
argument to the previous paragraph shows that
$x_2(t) = x_2(0) + F^+(t) - F^-(t)$ must converge.

\vspace{1mm}

No flow can ever leave node 2 for node 1, so we just need to bound the
flow from node 2 to node 3.
Flow leaves node 2 for node 3 at time $t$ if and only if
$x_4(t) \geq x_1(t)$. Since $x_1(t)\geq y_1$, it follows that
$\wn_2(t), \wn_3(t) \geq y_1$ as well.
Therefore, we can bound the non-negative flow from node 2 to node 3 as
follows:
\begin{align*}
\max(\flow{2}{3}(t),0)
& \leq  {\ip \cdot} \frac{x_2(t)\, x_3(t)}{\wn_2(t)\, \wn_3(t)} \cdot x_4(t)
\; \leq \; {\ip \cdot} \frac{x_2(t)\, x_3(t)}{y_1^2} \\
& \leq  {\ip \cdot}\frac{x_1(t)\, x_2(t)\, x_3(t)}{y_1^3}
\; \leq \;  \frac{1}{y_1^3} \cdot {\ip \cdot} \frac{x_1(t)\, x_2(t)\, x_3(t)}{\wn_1(t) \wn_2(t)}
\; = \; \frac{1}{y_1^3} \cdot \flow{1}{2}(t).
\end{align*}
Thus, the total positive flow from node 2 to node 3 is bounded above
by a constant times the total flow from node 1 to node 2,
which in turn is at most $x_1(0)$.
\end{enumerate}
\vspace{-5mm}
\end{proof}

\end{document}